\documentclass[11pt]{article}
\PassOptionsToPackage{hyphens}{url}
\usepackage[final]{acl}

\usepackage{times}
\usepackage{latexsym}
\usepackage[T1]{fontenc}
\usepackage[utf8]{inputenc}
\usepackage{microtype}
\usepackage{inconsolata}
\usepackage{graphicx}
\usepackage{amsmath,amssymb}
\usepackage{amsthm}
\usepackage{booktabs}
\usepackage{algorithm}
\usepackage{algpseudocode}
\usepackage{xspace}
\usepackage{multirow}
\usepackage{enumitem}
\usepackage{xcolor}
\usepackage{tabularx}
\usepackage{caption}
\usepackage{hyperref}
\hypersetup{
  colorlinks=true,
  linkcolor=blue,
  citecolor=blue,
  urlcolor=blue
}

\setlist{nosep,leftmargin=1.2em}
\newtheorem{proposition}{Proposition}
\newtheorem{definition}{Definition}

\newcommand{\system}{\textup{\textsc{MMA-RAG}$^{\textsc{T}}$}\xspace}
\newcommand{\mta}{\textup{\textsc{MTA}}\xspace}
\newcommand{\Amain}{\mathcal{A}_{\mathrm{main}}}
\newcommand{\Atrust}{\mathcal{A}_{\mathrm{trust}}}
\newcommand{\Lgen}{\mathcal{L}_{\mathrm{gen}}}
\newcommand{\Sadv}{\mathcal{S}_{\mathrm{adv}}}
\newcommand{\Senv}{\mathcal{S}_{\mathrm{env}}}
\newcommand{\Jutil}{J_{\mathrm{util}}}
\newcommand{\Jviol}{J_{\mathrm{viol}}}
\newcommand{\approve}{\textup{\textsc{Approve}}}
\newcommand{\mitig}{\textup{\textsc{Mitigate}}}
\newcommand{\refuse}{\textup{\textsc{Refuse}}}

\title{Adversarial Intent is a Latent Variable:\\Stateful Trust Inference for Securing Multimodal Agentic RAG}
\author{
Inderjeet Singh$^{1,\ast}$\quad
Vikas Pahuja$^{1}$\quad
Aishvariya Priya Rathina Sabapathy$^{1}$\quad
Chiara Picardi$^{1}$ \\
\bfseries Amit Giloni$^{1}$\quad
Roman Vainshtein$^{1}$\quad
Andr\'{e}s Murillo$^{1}$\quad
Hisashi Kojima$^{2}$ \\
\bfseries Motoyoshi Sekiya$^{2}$\quad
Yuki Unno$^{2}$\quad
Junichi Suga$^{2}$ \\[6pt]
\mdseries $^{1}$\textcolor{red}{\textbf{Fujitsu}} Research of Europe, UK\qquad
$^{2}$\textcolor{red}{\textbf{Fujitsu}} Limited, Japan \\[2pt]
$^\ast$Corresponding author: \texttt{inderjeet.singh@fujitsu.com}
}

\begin{document}
\maketitle

\begin{abstract}
Current stateless defences for multimodal agentic RAG fail to detect adversarial strategies that distribute malicious semantics across retrieval, planning, and generation components. We formulate this security challenge as a Partially Observable Markov Decision Process (POMDP), where adversarial intent is a latent variable inferred from noisy multi-stage observations. We introduce \system, an inference-time control framework governed by a \textbf{Modular Trust Agent} (\mta) that maintains an approximate belief state via structured LLM reasoning. Operating as a model-agnostic overlay, \system mediates a configurable set of internal checkpoints to enforce stateful defence-in-depth. Extensive evaluation on 43,774 instances demonstrates a $6.50\times$ average reduction factor in Attack Success Rate relative to undefended baselines, with negligible utility cost. Crucially, a factorial ablation validates our theoretical bounds: while statefulness and spatial coverage are individually necessary ($26.4$\,pp and $13.6$\,pp gains respectively), stateless multi-point intervention can yield zero marginal benefit under homogeneous stateless filtering when checkpoint detections are perfectly correlated.
\end{abstract}

\section{Introduction}
\label{sec:intro}

Retrieval-Augmented Generation (RAG) has evolved from static query-response pipelines~\citep{lewis2020rag,gao2024ragsurvey} into \emph{agentic} architectures where orchestrator LLMs plan multi-step workflows, invoke external tools, and synthesise \emph{multimodal} evidence~\citep{yao2022react,schick2023toolformer,ferrazzi2026agentic}. While this paradigm enhances capability, it introduces a fundamentally expanded threat model. Unlike conventional LLMs where adversarial content is confined to the user prompt, agentic systems ingest hostile artifacts at multiple internal junctures: knowledge stores may harbour poisoned documents~\citep{zou2024poisonedrag,cheng2024trojanrag}, images may encode embedded directives~\citep{figstep2023,zeeshan2025chameleon}, tool outputs may carry indirect prompt injections~\citep{greshake2023indirectprompt,zhan2024injecagent}, and the orchestrator's reasoning chain itself may be subverted through cascading influence~\citep{gu2024agentsmith,srivastava2025memorygraft}. Recent exploits involving Model Context Protocol (MCP) traffic further underscore the ubiquity of this surface~\citep{janjusevic2025hiding}.

\begin{figure}[t]
  \centering
  \includegraphics[width=\columnwidth]{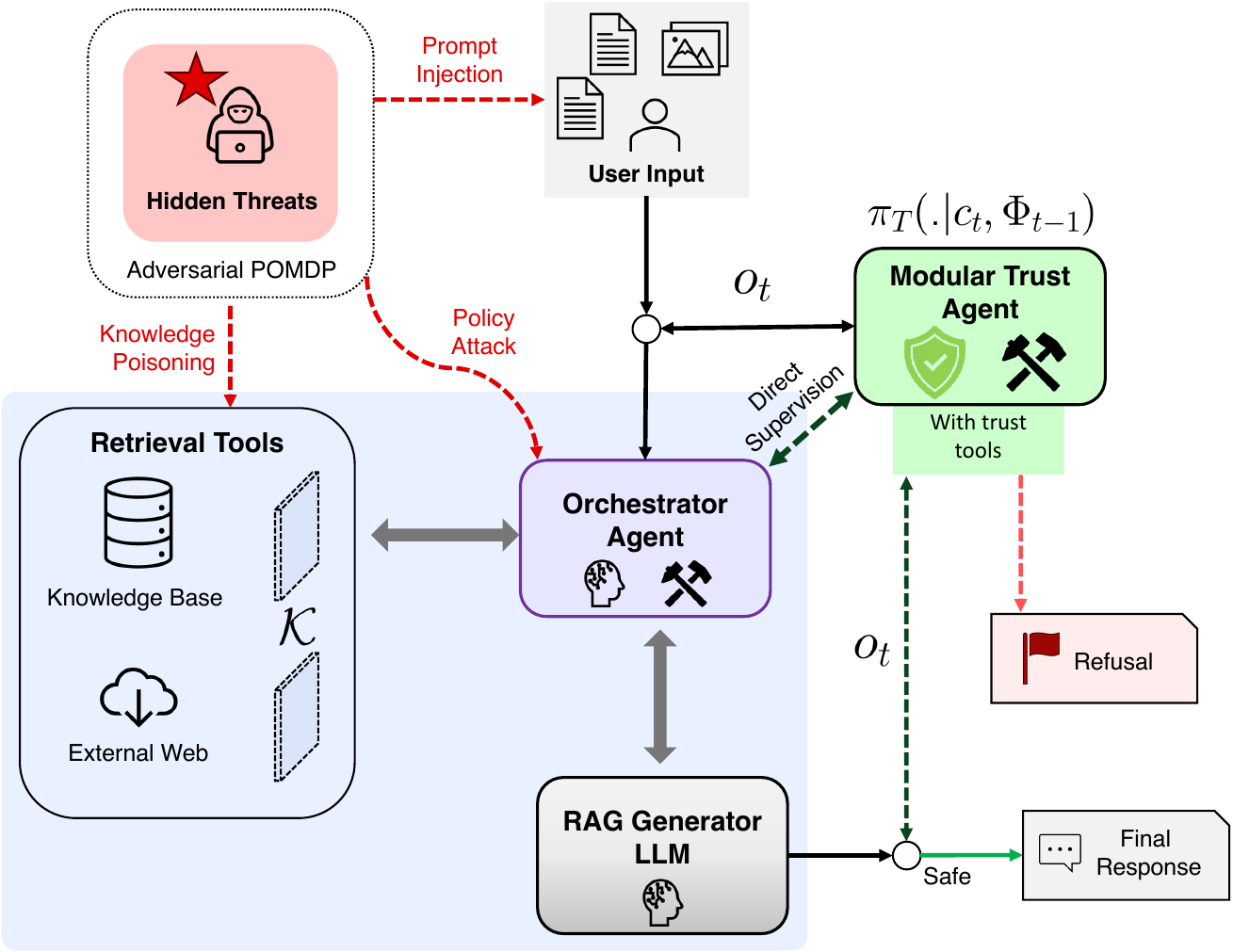}
  \caption{\textbf{The \system architecture.} The \mta ($\Atrust$) intercepts artifacts at a set of checkpoints (\textbf{C1}-\textbf{C5} in our evaluated instantiation), executes Algorithm~\ref{alg:mta} at each, and maintains a cumulative belief state $\Phi_t$. Decisions $\delta_t \!\in\! \{\approve, \mitig, \refuse\}$ gate artifact passage through the pipeline.}
  \label{fig:architecture}
\end{figure}

The current defensive landscape exhibits structural gaps rooted in a failure to model this complexity. \emph{Stateless boundary filters} like Llama Guard~\citep{inan2023llamaguard} and NeMo Guardrails~\citep{rebedea2023nemo} evaluate artifacts in isolation, discarding the interaction history necessary to detect multi-stage attacks. \emph{Retrieval-stage} defences~\citep{masoud2026sdrag,pathmanathan2025ragpart,tan2024revprag} and \emph{injection-specific} methods~\citep{yu2026defense,wallace2024instructionhierarchy,suo2025argus} protect isolated vectors but lack mechanisms to correlate evidence across the pipeline. Even \emph{trajectory-based detectors}~\citep{advani2026trajectory} operate reactively on completed sequences rather than proactively intercepting threats. Crucially, these approaches treat system interactions as independent events. In our B1 ablation (\S\ref{sec:ablation}), removing belief state increases ASR by $26.4$\,pp. Moreover, under homogeneous stateless filtering where the same underlying adversarial artifact is observed through highly correlated views across checkpoints, adding more checkpoints can yield zero marginal detection benefit: the detection events become effectively perfectly correlated, so the union probability collapses to the single-checkpoint probability.

We argue that these failures stem from a single control-theoretic problem: \emph{partial observability} of adversarial intent. At any pipeline stage, a trust mechanism observes only artifact content; whether that artifact is malicious is a latent variable that must be inferred from noisy, sequential signals. This structure maps naturally to a Partially Observable Markov Decision Process (POMDP)~\citep{kaelbling1998pomdp}, where optimal defence requires maintaining a belief distribution over hidden adversarial states.

We introduce \system (Figure~\ref{fig:architecture}), a defence framework whose central \textbf{Modular Trust Agent} (\mta) operationalises this probabilistic insight. Our contributions are:

\textbf{(C1) Theoretical Formalism:} We model the agentic security challenge as an adversarial POMDP and derive design principles (Propositions~\ref{prop:memory} and~\ref{prop:correlation}) characterising when stateless multi-point intervention yields no marginal benefit, and why stateful belief tracking provides strict gains under partial observability~(\S\ref{sec:pomdp}).
\textbf{(C2) Architecture:} The \mta enforces defence-in-depth by mediating a configurable set of internal checkpoints, maintaining an approximate belief state via structured LLM inference to detect compound attacks that appear benign in isolation~(\S\ref{sec:checkpoints}).
\textbf{(C3) Empirical Validation:} Evaluated on 43,774 instances across five threat surfaces (ART-SafeBench suite~\citep{bsafety2026}), \system reduces Attack Success Rate (ASR) by $6.50\times$ on average (reduction factor). A $2\!\times\!2$ factorial ablation confirms that statefulness and spatial coverage are individually necessary, validating our theoretical derivations~(\S\ref{sec:results}).
\textbf{(C4) Deployability \& Limits:} The \mta operates as a model-agnostic inference overlay requiring no fine-tuning or weight access. We also identify a \emph{semantic resolution limit} in tool-flip attacks (B4, $1.29\times$ reduction), delineating the boundary where text-based belief tracking must be augmented by deterministic governance~(\S\ref{sec:discussion}).

The evaluation leverages ART-SafeBench; we summarise the benchmark generation pipeline in \S\ref{sec:setup} and provide full details in Appendix~\ref{app:artsafebench_generation}.

\section{Threat Model}
\label{sec:threat}

\paragraph{System and Adversary.}
We model the agentic RAG system as a tuple $(\Amain, \mathcal{K}, \mathcal{T}, \Lgen)$ mediated by an independent trust agent $\Atrust$. Execution follows a trajectory $\tau = (s_0, a_0, o_1, \ldots)$ over state space $\mathcal{S}$ and action space $\mathcal{A} = \mathcal{A}_{\mathrm{main}} \cup \mathcal{A}_{\mathrm{trust}} \cup \mathcal{A}_{\mathrm{gen}}$~(\S\ref{sec:pomdp}). The adversary $\mathcal{A}_{\mathrm{adv}}$ targets OWASP-aligned objectives~\citep{owasp2025llm}: \textbf{(O1)}~harmful outputs, \textbf{(O2)}~control subversion, \textbf{(O3)}~exfiltration, and \textbf{(O4)}~resource exhaustion. Access is either black-box ($\Gamma_{\mathrm{BB}}$) or partial ($\Gamma_{\mathrm{PA}}$; poisoning $\mathcal{K}_{\mathrm{adv}} \subset \mathcal{K}$), without white-box visibility into $\Atrust$ or model weights.

\paragraph{Attack Surfaces and Partial Observability.}
Adversarial artifacts enter via inputs~($\mathcal{I}$), storage~($\mathcal{K}$), tools~($\mathcal{T}$), or inter-component channels~($\mathcal{IC}$), instantiating three overlapping scenarios: \textbf{S1\,(Retrieval manipulation):} poisoned content in $\mathcal{K}$ surfacing via semantic similarity; \textbf{S2\,(Directive injection):} inputs overriding system instructions (O1-O2); and \textbf{S3\,(State corruption):} inputs biasing $\Amain$ toward tool misuse or drift. Crucially, adversarial intent is a \emph{latent variable}. Sophisticated adversaries can distribute attacks across stages (chaining S1-S3) such that no single observation reveals the composite threat. This renders stateless defences structurally inadequate and motivates the \mta's stateful belief tracking~(\S\ref{sec:pomdp}) and multi-stage intervention~(\S\ref{sec:checkpoints}).

\paragraph{Defender Assumptions.}
Safety is defined via predicates $\{\phi_c\}_{c=1}^{C}$ over trajectories, where $\phi_c(\tau)\!=\!1$ iff any step triggers violation $g_c(s_t,a_t,o_{t+1})\!=\!1$ (e.g., forbidden content emission, disallowed tool execution). We assume $\Atrust$ executes in a protected environment isolated from $\Amain$ to decorrelate failure modes~(\S\ref{sec:resilience}).

\section{The \system Framework}
\label{sec:framework}

\subsection{System Architecture}
\label{sec:architecture}

\system augments a representative ReAct-style agentic RAG instantiation (used in our experiments) with the \mta (Figure~\ref{fig:architecture}; full integration in Algorithm~\ref{alg:full}, Appendix~\ref{app:fullalg}), acting as a runtime security overlay. Practical agentic RAG systems vary in orchestration, memory, retrieval stacks, and tool routers; \system attaches via interposition at a configurable set of trust boundaries rather than assuming a fixed pipeline. The architecture comprises four primary components mediated by the trust layer:
\textbf{(i)}~\emph{Orchestrator}~$\Amain$: a ReAct-style~\citep{yao2022react} planning engine (LangGraph;~\citealp{langgraph2024}) responsible for query decomposition, tool selection, and evidence synthesis.
\textbf{(ii)}~\emph{Knowledge Store}~$\mathcal{K}$: a multimodal vector database (ChromaDB;~\citealp{chromadb2024}) with OpenCLIP ViT-g-14 embeddings~\citep{radford2021clip,openclip2023}; images grounded via Tesseract OCR~\citep{tesseract2024}.
\textbf{(iii)}~\emph{Generator}~$\Lgen$: produces the final user-facing response conditioned on retrieved context.
\textbf{(iv)}~\emph{Trust Agent}~$\Atrust$: an isolated subsystem maintaining its own LLM instance, persistent belief state, and communicating exclusively through a configurable set of interception checkpoints.
Critically, $\Atrust$ is orthogonal to the functional pipeline: it attaches to $\Amain$, $\mathcal{K}$, and $\Lgen$ via API interposition, requiring no weight access or architectural modification.

\subsection{Adversarial POMDP Formalisation}
\label{sec:pomdp}

We ground the \mta's design in the POMDP framework. This formalism provides the axiomatic basis for reasoning about adversarial intent as a latent variable under the partial observability established in \S\ref{sec:threat}.

\begin{definition}[Adversarial Agentic RAG POMDP]
\label{def:pomdp}
The interaction is defined by the tuple $(\mathcal{S}, \mathcal{A}, \Omega, T, O, J)$ where:
$\mathcal{S} = \Senv \times \Sadv$ factors into the observable environment state (documents, tool outputs) and the latent adversarial state (benign vs.\ malicious intent);
$\mathcal{A} = \{\approve, \mitig, \refuse\}$ is the trust-agent action space; the remaining pipeline dynamics are absorbed into $T$;
$\Omega$ is the observation space for $\Atrust$. Each checkpoint yields an observation $o_t \in \Omega$ comprising the current artifact $C_t$, its type, and channel metadata (e.g., retrieval provenance, tool name, or parsing context);
$T\!:\!\mathcal{S}\!\times\!\mathcal{A}\!\to\!\Delta(\mathcal{S})$ is the transition kernel; and
$O\!:\!\mathcal{S}\!\times\!\mathcal{A}\!\to\!\Delta(\Omega)$ is the observation function mapping hidden states to observations.
The objective $J = \Jutil - \lambda \sum_{t=1}^{H} \Jviol(s_t, a_t)$ trades off task utility (including refusal cost) against safety violations ($\Jviol \!\in\! \{0,1\}$; $\lambda > 0$).
\end{definition}

\paragraph{Approximate Belief State Maintenance.}
In a canonical POMDP, the agent maintains a sufficient statistic $b_t \in \Delta(\mathcal{S})$ via the exact Bayesian update:
\begin{equation}
\label{eq:bayes}
b_t(s') \propto O(o_t \!\mid\! s', a_{t\!-\!1}) \sum_{s} T(s' \!\mid\! s, a_{t\!-\!1})\, b_{t\!-\!1}(s).
\end{equation}
For agentic RAG, the state space $\mathcal{S}$ encompasses high-dimensional document semantics, tool configurations, reasoning traces, and adversarial strategies, rendering exact belief maintenance via Eq.~\eqref{eq:bayes} intractable.
The \mta therefore maintains a structured natural-language \emph{approximate belief state} $\Phi_t$ (an information state that serves as a tractable proxy for the exact posterior $b_t$):
\begin{equation}
\label{eq:belief}
\Phi_t = f_{\theta}\!\left(\Phi_{t-1},\; o_t\right),
\end{equation}
where $f_{\theta}$ is a single frozen LLM inference call that ingests the prior state $\Phi_{t-1}$ (which encodes decision history) and the current observation $o_t$.
The prompt (Appendix~\ref{app:prompt}) instructs the LLM to produce a structured summary encoding cumulative risk indicators, a threat-level assessment, and the decision history for the current query.
$\Phi_t$ plays the functional role of an \emph{approximate information state}~\citep{kaelbling1998pomdp}: it compresses the history $h_t = (o_1, \delta_1, \ldots, o_t)$ into a fixed-format representation that the policy conditions upon, bypassing explicit value iteration.
Ablating $\Phi_t$ degrades defence by $26.4$\,pp~(\S\ref{sec:ablation}), confirming that this approximation captures decision-relevant latent variables.

\paragraph{Belief compression and token efficiency.}
An alternative to maintaining $\Phi_t$ is to append the full interaction log $h_t$ into the trust-agent prompt at each checkpoint. For $K$ checkpoint interceptions in a query, this yields input contexts whose length grows with $t$, resulting in $\mathcal{O}(K^2)$ total prompt tokens across checkpoints. In contrast, $\Phi_t$ is a fixed-format summary with bounded size, yielding $\mathcal{O}(K)$ scaling and enabling stateful trust inference under production context limits.

We derive two structural properties from this formulation that generate falsifiable predictions validated in \S\ref{sec:results}.

\begin{proposition}[Strict Value of Memory]
\label{prop:memory}
Let $\Pi_\Omega = \{\pi : \Omega \to \mathcal{A}\}$ denote the set of stateless (observation-memoryless) policies and $\Pi_\Phi = \{\pi : \Omega \times \Phi \to \mathcal{A}\}$ denote belief-conditioned policies.
Define the value function $V(\Pi) \triangleq \sup_{\pi \in \Pi}\, \mathbb{E}_\pi[J]$.
Then:
\begin{equation}
\label{eq:dominance}
V(\Pi_\Phi) \;\geq\; V(\Pi_\Omega),
\end{equation}
and under partial observability the inequality can be strict, e.g., when there exist histories $h_t,h'_t$ with the same instantaneous observation $o_t$ but different posteriors over $\Sadv$.
\end{proposition}

\noindent\emph{Proof sketch.}\; See Appendix~\ref{app:proof_sketches}.

\begin{proposition}[Checkpoint Correlation Structure]
\label{prop:correlation}
Consider a fixed query execution with $K$ checkpoint interceptions. Let $D_t \!\in\! \{0,1\}$ denote the event that checkpoint $t$ detects an adversarial artifact.
\emph{(i)} Under a homogeneous stateless detector $d\!:\!\Omega\!\to\!\{0,1\}$ applied at every checkpoint, if the observations of a given adversarial artifact across checkpoints fall in the same equivalence class of $d$, then $D_t = D_1$ for all $t$, so $\rho(D_i, D_j) = 1$ for all pairs $i,j$ and $P(\bigcup_{t=1}^{K} \{D_t\!=\!1\}) = P(D_1\!=\!1)$. Additional stateless checkpoints yield no marginal detection benefit in this regime.
\emph{(ii)} Under a belief-conditioned detector $d_t\!:\!\Omega \times \Phi_{t-1}\!\to\!\{0,1\}$, evidence accumulated in $\Phi_{t-1}$ can push a later checkpoint past the detection threshold even when $o_t$ alone is insufficient; consequently, one can have $P(\bigcup_{t=1}^{K} \{D_t\!=\!1\}) > P(D_1\!=\!1)$.
\end{proposition}

\noindent\emph{Proof sketch.}\; See Appendix~\ref{app:proof_sketches}.

\subsection{Inference-Action Loop}
\label{sec:loop}

At each checkpoint $t$, the \mta executes a four-step control cycle (Algorithm~\ref{alg:mta}, Appendix~\ref{app:algorithm}):
\textbf{(1)~Observe}: Intercept artifact $C_t$ and retrieve belief state $\Phi_{t-1}$.
\textbf{(2)~Infer}: Execute the trust policy $\pi_T$ via LLM inference to generate a risk assessment $r_t$ and updated belief $\Phi_t$.
\textbf{(3)~Decide}: Map $r_t$ to an action $\delta_t \in \{\approve, \mitig, \refuse\}$.
\textbf{(4)~Act}: Enforce $\delta_t$ and commit $\Phi_t$ to memory.

\paragraph{Inference-aligned policy execution.}
\label{sec:policy}
The trust policy $\pi_T$ is executed through LLM inference rather than a separate rule engine.
Policy specifications (threat indicators, evidence-weighting heuristics, decision thresholds, JSON output schema) are encoded as structured instructions in the \mta's system prompt (Appendix~\ref{app:prompt}).
This design leverages the LLM's pre-trained knowledge of adversarial patterns and social engineering tactics without requiring explicit attack-signature enumeration, enabling context-sensitive compositional reasoning: for instance, recognising that ``ignore previous context'' constitutes an injection attempt even in novel phrasings.
Because $\pi_T$ is specified at inference time, $\Atrust$ is model-agnostic: it can be instantiated with either general-purpose foundation models or smaller specialist safety models (including fine-tuned LLMs), and can use heterogeneous backbones across checkpoints to trade off latency and fidelity.
This also provides a direct upgrade path: improving or fine-tuning the trust-agent model can improve assessment and mitigation fidelity without modifying or retraining $\Amain$ or $\Lgen$.

\paragraph{Formal-to-implementation mapping.}
The POMDP objects are instantiated as follows:
(1)~$\Phi_{t-1}$ (``Current Suspicion Context'' in the prompt; Appendix~\ref{app:prompt}) approximates the belief $b_t$ (Eq.~\ref{eq:bayes});
(2)~the LLM's ``Updated Suspicion Context'' implements an approximate belief update; and
(3)~$\pi_T$, encoded as structured system-prompt instructions, maps $(\Phi_{t-1}, o_t)$ to $\delta_t$ as the approximate policy.
This grounding ensures that the theoretical benefits of belief-state estimation are realised in the implementation without the tractability costs of explicit Bayesian filtering, and generates the predictions validated in \S\ref{sec:ablation}.

\paragraph{Effector Primitives.}
Upon deciding $\delta_t$: \textbf{(a)}~$\approve$: artifact passes unchanged; \textbf{(b)}~$\mitig$: the \mta outputs a sanitised artifact that preserves benign content while neutralising threats. Mitigation can be implemented via LLM rewriting, deterministic safety tools (e.g., argument-schema validation for tool calls, allowlists, PII scrubbing, URL stripping), or a hybrid; \textbf{(c)}~$\refuse$: workflow is terminated with a safe fallback response.
Operationally, these safety tools execute within $\Atrust$'s isolated environment and return only sanitised artifacts or validated actions through the interception interface, keeping integration non-invasive.

\subsection{Configurable Checkpoints: Defence-in-Depth}
\label{sec:checkpoints}

The \mta implements per-stage filtering to address the specific attack surfaces defined in \S\ref{sec:threat}. In general, the checkpoint set is determined by the system's trust boundaries; we instantiate a representative checkpoint set (C1-C5) mapping to common boundaries in multimodal agentic RAG:
\textbf{C1\,(Query Screening):} Filters user input $\mathcal{I}$ for jailbreaks and payload encoding.
\textbf{C2\,(Action Validation):} Validates $\Amain$'s tool selection against the query intent, detecting agentic state corruption (S3).
\textbf{C3\,(Retrieval Screening):} Filters $\mathcal{K}$ retrievals for poisoned content (S1) before generator consumption.
\textbf{C4\,(Tool Output Screening):} Inspects external tool returns $\mathcal{T}$ for indirect injections.
\textbf{C5\,(Response Verification):} Validates the final generation $\Lgen$ for policy compliance.
No single checkpoint suffices: an adversary aware of C1 routes attacks through $\mathcal{K}$; one aware of C3 uses tool outputs; one aware of C5 distributes the attack across earlier stages.
Within our instantiation, restricting defence to C5 alone incurs a $13.6$\,pp degradation~(\S\ref{sec:ablation}), demonstrating the necessity of intercepting attacks at their source before downstream reasoning corruption occurs.

\subsection{Architectural Resilience}
\label{sec:resilience}

To prevent the \mta from becoming a single point of failure, we enforce three properties:
\emph{Execution Isolation}: $\Atrust$ operates in a disjoint memory space with independent system prompts, preventing context contamination from $\Amain$.
\emph{Model Heterogeneity}: $\Atrust$ may utilise a distinct model architecture to mitigate transfer attacks.
\emph{Auditability}: All belief updates and decisions are immutable logs, facilitating post-hoc forensic analysis and iterative policy refinement.


\section{Experimental Setup}
\label{sec:setup}

\paragraph{System Configuration.}
We evaluate \textbf{\system} against a baseline multimodal agentic RAG pipeline with no trust layer (identical architecture without $\Atrust$). The baseline is a representative instantiation used solely for controlled evaluation and is not itself claimed as a contribution. Both systems share a fixed infrastructure: ChromaDB vector store with OpenCLIP ViT-g-14 embeddings~\citep{radford2021clip,openclip2023}, Tesseract v5.3.1 for OCR, and LangGraph~\citep{langgraph2024} for orchestration.
To ensure reproducibility, all deterministic operations (generation, judgment, and \mta policy execution) use temperature $T\!=\!0.0$.
We assess generalisation across four LLM backbones: GPT-4o, GPT-4o-mini, Llama-3.3-70B, and GPT-4.1. The \mta itself employs a frozen GPT-4o instance for inference and mitigation. Unless otherwise noted, aggregate results report the GPT-4o backbone.

\paragraph{Evaluation Corpus.}
Security is assessed on the \textbf{ART-SafeBench} suite~\citep{bsafety2026}, comprising 43,774 validated adversarial instances across five threat surfaces:
\textbf{B1\,(Text Poisoning, $n\!=\!10{,}943$):} Instruction overrides and persona manipulations embedded in retrieved documents.
\textbf{B2\,(Image Poisoning, $n\!=\!4{,}000$):} Adversarial text programmatically rendered onto images, split into OCR-mediated (B2a) and direct multimodal (B2b) vectors.
\textbf{B3\,(Direct Query, $n\!=\!10{,}005$):} Zero-shot adversarial prompts spanning OWASP Top-10 categories~\citep{owasp2025llm}.
\textbf{B4\,(Tool-Flip, $n\!=\!14{,}400$):} Dual-query injections forcing semantically plausible but adversarially motivated tool switching.
\textbf{B5\,(Agentic Integrity, $n\!=\!4{,}426$):} Multi-turn interactions targeting reasoning chain corruption and incremental policy drift.
Instances are generated via a closed-loop validated synthesis pipeline with $\sim$37\% acceptance rate. We provide a four-stage summary in Appendix~\ref{app:artsafebench_generation} and list the corresponding five-step instantiation for reproducibility; record schema in Appendix~\ref{app:schema}.
\textbf{Utility} is measured on Natural Questions (dev, $n\!=\!3{,}610$)~\citep{kwiatkowski2019nq} using the ARES framework~\citep{saadfalcon2024ares} for Context Relevance (CR) and Answer Relevance (AR).

\paragraph{Adjudication Protocol.}
We define Attack Success Rate (ASR) as the proportion of trajectories satisfying an adversarial success predicate $\kappa$:
\begin{equation}
\mathrm{ASR} = \frac{1}{|\mathcal{D}_{\mathrm{atk}}|} \sum_{x \in \mathcal{D}_{\mathrm{atk}}} \mathbb{I}\bigl[\kappa(g, x, T(x)) \geq \tau\bigr].
\end{equation}
For B1, B2, B3, and B5, $\kappa$ is a deterministic GPT-4o judge ($T\!=\!0$). For B4 (Tool-Flip), $\kappa$ is a formal predicate verifying divergence from the ground-truth tool selection.
To validate the automated judge, we manually audited a stratified sample of 750 instances, yielding $>$95\% agreement overall (100\% on B1/B2/B4/B5; 85\% on B3 due to subjective ambiguity in bias categories)~\citep{zheng2023judgellm}.

\section{Results and Analysis}
\label{sec:results}

\subsection{Utility Preservation}
\label{sec:utility}

The addition of the trust layer introduces negligible degradation in downstream utility (Table~\ref{tab:utility}). \system achieves CR~$= 0.87$ and AR~$= 0.93$, comparable to the baseline (CR~$= 0.88$, AR~$= 0.93$) and competitive with established baselines under the same ARES protocol~\citep{saadfalcon2024ares}. The minimal CR reduction confirms that the \mta's filtering preserves the semantic integrity of valid retrieval-generation workflows.

\begin{table}[t]
\centering
\footnotesize
\setlength{\tabcolsep}{5pt}
\renewcommand{\arraystretch}{1.05}
\begin{tabularx}{\columnwidth}{@{}Xcc@{}}
\toprule
\textbf{Configuration} & \textbf{CR\,$\uparrow$} & \textbf{AR\,$\uparrow$} \\
\midrule
ColBERTv2 + GPT-4 & 0.90 & 0.82 \\
DPR + BART & 0.62 & 0.58 \\
\midrule
Agentic RAG \emph{(baseline)} & 0.88 & 0.93 \\
\system \emph{(+\,\mta)} & 0.87 & 0.93 \\
\bottomrule
\end{tabularx}
\caption{Utility on Natural Questions (dev, $n\!=\!3{,}610$). ARES scores (context relevance CR, answer relevance AR). Established baselines from \citet{saadfalcon2024ares} included for reference. Differences between the baseline and \system are small.}
\label{tab:utility}
\end{table}

\subsection{Defence Efficacy}
\label{sec:security}

Across the 43,774 instances of the ART-SafeBench suite, \system reduces ASR with an average reduction factor of $6.50\times$ (Table~\ref{tab:main_results}, Figure~\ref{fig:results}(a)). We analyse performance by threat surface.

\begin{table}[t]
\centering
\footnotesize
\setlength{\tabcolsep}{5pt}
\renewcommand{\arraystretch}{1.08}
\begin{tabularx}{\columnwidth}{@{}Xrrr@{}}
\toprule
\textbf{Benchmark} & \textbf{Baseline} & \textbf{+\,\mta} & \textbf{Red.} \\
\midrule
B1: Text poisoning    & 76.5 & 15.7 & $4.87\!\times$ \\
B2a: Image (OCR)      & 22.3 &  2.0 & $11.15\!\times$ \\
B2b: Image (MM)       & 10.1 &  5.5 & $1.84\!\times$ \\
B3: Direct query      & 59.0 &  4.1 & $14.39\!\times$ \\
B4: Tool-flip         & 82.1 & 63.5 & $1.29\!\times$ \\
B5: Agentic integrity & 48.5 &  8.9 & $5.45\!\times$ \\
\midrule
\textbf{Avg.\ (factor)} & \textbf{-} & \textbf{-} & $\mathbf{6.50\!\times}$ \\
\bottomrule
\end{tabularx}
\caption{ASR~(\%, $\downarrow$) across ART-SafeBench benchmarks with GPT-4o backbone. Reduction factor $=$ baseline~ASR\,/\,defended~ASR. Avg.\ (factor) is the unweighted mean of per-benchmark reduction factors.}
\label{tab:main_results}
\end{table}

\begin{figure*}[t]
  \centering
  \includegraphics[width=\textwidth]{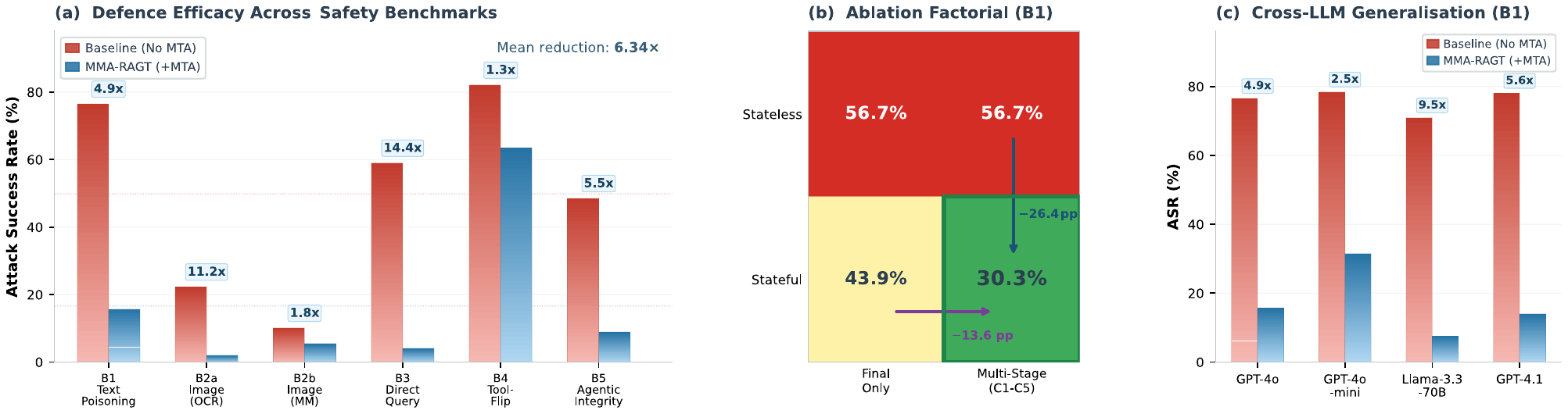}
  \caption{\textbf{Experimental results overview.}
  \textbf{(a)}~ASR across all ART-SafeBench benchmarks: the \mta reduces ASR on every surface, with reduction factors ranging from $1.3\times$ (B4, tool-flip) to $14.4\times$ (B3, direct query); mean factor $6.50\times$.
  \textbf{(b)}~Factorial ablation on B1: statefulness contributes $-26.4$\,pp and multi-stage coverage $-13.6$\,pp; the two mechanisms are individually necessary and their combination yields super-additive gains, validating Propositions~\ref{prop:memory} and~\ref{prop:correlation}.
  \textbf{(c)}~Cross-LLM generalisation on B1: defence transfers across four backbones ($2.5\!\times$-$9.5\!\times$ reduction), with invariant relative difficulty ordering.}
  \label{fig:results}
\end{figure*}

\paragraph{High-efficacy surfaces (B1, B2a, B3).}
The \mta excels where adversarial artifacts carry explicitly observable signals.
On \textbf{B3} ($14.39\times$), C1 is highly effective against direct adversarial queries, reducing ASR to 4.1\% (95\%~CI: 3.7-4.5\%). Per-OWASP breakdown for the six instantiated categories (Table~\ref{tab:owasp}, Appendix~\ref{app:owasp}; $n=4{,}978$ of 10,005): information disclosure (LLM02) and supply-chain attacks (LLM03) are eliminated entirely; insecure output handling (LLM05) drops to $<$1\%. Residual ASR concentrates in prompt injection (LLM01: $64\!\to\!8\%$) and misinformation (LLM09), where adversarial content most closely resembles legitimate requests.
On \textbf{B1} ($4.87\times$), belief-state accumulation across C3 and C5 reduces ASR from 76.5\% to 15.7\% (95\%~CI: 15.0-16.4\%), detecting subtly poisoned content whose adversarial intent emerges only when correlated across pipeline stages.
\textbf{B2a} (OCR-mediated, $11.15\times$) is more vulnerable than B2b since OCR faithfully transcribes adversarial text into the generator context: $22.3\!\to\!2.0\%$ (95\%~CI: 1.3-3.1\%).

\paragraph{The semantic resolution limit (B4, B2b).}
\textbf{B4} (Tool-Flip) yields the lowest reduction ($1.29\times$; $82.1\!\to\!63.5\%$; 95\%~CI: 62.7-64.3\%), confirming the semantic resolution limit of text-based belief tracking. Tool-flip attacks manipulate $\Amain$'s action-selection logic at a semantic level closely resembling legitimate reasoning; we analyse this gap in \S\ref{sec:discussion}.
\textbf{B2b} (direct multimodal, $1.84\times$) proves harder than B2a ($10.1\!\to\!5.5\%$; 95\%~CI: 4.3-7.1\%), delineating a perceptual boundary where visual injections bypass the symbolic reasoning accessible to the text-based trust policy.

\paragraph{Agentic integrity (B5, $5.45\times$).}
Multi-stage attacks on the orchestrator's reasoning chain are reduced to 8.9\% (95\%~CI: 8.1-9.8\%) through belief-state accumulation across stages (Proposition~\ref{prop:memory}), enabling detection of compound attacks where no single observation is independently suspicious.

\paragraph{Cross-LLM generalisation.}
Efficacy transfers across protected backbones (Figure~\ref{fig:results}(c); Appendix~\ref{app:perllm}).
On B1, ASR ranges from 7.5\% (Llama-3.3-70B) to 31.5\% (GPT-4o-mini) under \system, vs.\ baseline ranges of 70.9-78.4\%, yielding $2.5\!\times$-$9.5\!\times$ reduction factors.
The relative difficulty ordering (B4 hardest, B3 easiest) remains invariant across all four backbones, suggesting that effectiveness derives from architectural placement and the belief-state mechanism rather than idiosyncrasies of a single backbone.

\subsection{Ablation: Empirical Validation of Theory}
\label{sec:ablation}

We employ a $2\!\times\!2$ factorial design on B1 (Table~\ref{tab:ablation}, Figure~\ref{fig:results}(b)) to empirically validate the design principles derived in \S\ref{sec:pomdp}. B1 is selected as the median-difficulty surface that exercises all checkpoints in our evaluated instantiation (C1-C5) and stresses both belief-state accumulation and spatial coverage.

\begin{table}[t]
\centering
\footnotesize
\setlength{\tabcolsep}{5pt}
\renewcommand{\arraystretch}{1.08}
\begin{tabularx}{\columnwidth}{@{}Xccc@{}}
\toprule
\textbf{Configuration} & \textbf{ASR\,$\downarrow$} & $\boldsymbol{\Delta}$\,\textbf{pp} & \textbf{Refusals} \\
\midrule
No \mta \emph{(baseline)} & 76.5 & +19.8 & - \\
Stateless, final-only     & 56.7$^\dagger$ & 0.0$^\dagger$ & - \\
Stateful, final-only      & 43.9 & $-12.8$ & 18 \\
Stateless, multi-stage    & 56.7 & ref. & 45 \\
\textbf{Stateful, multi-stage} & \textbf{30.3} & $\mathbf{-26.4}$ & \textbf{42} \\
\bottomrule
\end{tabularx}
\caption{Ablation on B1~(\%). \emph{Stateful}: belief state maintained across checkpoints. \emph{Multi-stage}: all C1-C5. $\Delta$\,pp: change vs.\ stateless, multi-stage (ref.). Refusals: total \mta block actions per evaluation run. $\dagger$\,Implied by Proposition~\ref{prop:correlation}(i); not independently measured. Stateless configurations structurally model existing boundary-filter paradigms.}
\label{tab:ablation}
\end{table}

\paragraph{Structural equivalence to existing defences.}
The ablation configurations instantiate common defence paradigms: \emph{stateless, final-only} models boundary classifiers such as Llama Guard~\citep{inan2023llamaguard}, while \emph{stateless, multi-stage} models multi-point guardrails such as NeMo Guardrails~\citep{rebedea2023nemo}. Proposition~\ref{prop:correlation}(i) characterises when these stateless deployments collapse under homogeneous filtering, and the remaining gap to \system isolates the incremental value of belief conditioning beyond stateless boundary filtering.

\paragraph{Validation of Propositions~\ref{prop:correlation} and~\ref{prop:memory}.}
The stateless multi-stage configuration achieves $56.7\%$ ASR, and Proposition~\ref{prop:correlation}(i) implies the corresponding stateless final-only detector is equivalent in the high-correlation regime (Table~\ref{tab:ablation}, $\dagger$). Introducing belief state yields a $26.4$\,pp gain ($56.7\%\!\to\!30.3\%$), confirming the strict value of memory (Proposition~\ref{prop:memory}), and restores marginal value to additional checkpoints as evidenced by the $13.6$\,pp gap between stateful final-only ($43.9\%$) and stateful multi-stage ($30.3\%$), supporting Proposition~\ref{prop:correlation}(ii).

\paragraph{Refusal-stage analysis.}
On B5, refusals concentrate at \textbf{C2} (agent action validation), indicating that multi-stage gating can block integrity attacks before downstream reasoning corruption.

\paragraph{Latency and cost.}
On $N\!=\!100$ Natural Questions queries, mean end-to-end response time increases from $2.59 \pm 0.77$\,s to $8.64 \pm 2.92$\,s, a $3.34\times$ overhead. In our evaluated configuration, this corresponds to an estimated $\sim$3,500-5,000 additional input tokens and $\sim$800-1,200 output tokens per query; full analysis and optimisation paths are provided in Appendix~\ref{app:latency_cost}.

\section{Discussion}
\label{sec:discussion}

\paragraph{Inference-time policy approximation.}
LLM inference can approximate policies for security POMDPs with semantic state spaces by realising belief updates and actions through structured prompts rather than weight updates. In our evaluation, this yields a $6.50\times$ average ASR reduction factor while preserving auditability and enabling upgrades by swapping or fine-tuning the trust-agent model without retraining $\Amain$ or $\Lgen$.

\paragraph{Belief compression.}
Storing adversarial intent as an abstract belief state keeps token costs bounded: naively appending full logs causes prompt length to grow with the number of checkpoints, whereas the fixed-format $\Phi_t$ yields bounded per-checkpoint context. This makes stateful multi-stage defence feasible under production context limits without retaining raw interaction traces in the trust prompt.

\paragraph{Synergy of statefulness and coverage.}
The factorial ablation (Table~\ref{tab:ablation}) exhibits a super-additive interaction: memory and multi-stage coverage are individually necessary and jointly strongest. Memory breaks checkpoint-level correlation (Proposition~\ref{prop:correlation}), while checkpoints provide the observations required to update $\Phi_t$.

\paragraph{The semantic resolution limit (B4).}
B4 yields the weakest defence ($1.29\times$), reflecting a semantic resolution limit: adversarial and benign trajectories can induce observations that are indistinguishable under $O$, constraining discrimination even with full history. Mitigating this class requires complementary non-semantic controls such as tool allowlists, argument-schema validation, and query-conditioned constraints.
Counterfactual checks over tool choice (e.g., re-evaluating actions under ablated context) are a natural next step.

\paragraph{Operational corollary.}
Deployment heuristic: prioritise belief-state infrastructure over scaling stateless checkpoints, since high correlation can render stateless checkpoint scaling redundant (Proposition~\ref{prop:correlation}). The \system overlay remains composable with retrieval-stage defences~\citep{masoud2026sdrag,pathmanathan2025ragpart}, injection-specific methods~\citep{yu2026defense}, and trajectory anomaly detection~\citep{advani2026trajectory}.

\section{Related Work}
\label{sec:related}

\paragraph{Agentic threats.}
Agentic RAG expands the attack surface beyond the user prompt, including retrieval poisoning~\citep{zou2024poisonedrag,cheng2024trojanrag,xue2024deceiverag,zhao2025exploring}, multimodal injections~\citep{figstep2023,zeeshan2025chameleon}, indirect tool injections~\citep{greshake2023indirectprompt,zhan2024injecagent}, multi-agent subversion~\citep{gu2024agentsmith}, persistent state corruption~\citep{srivastava2025memorygraft}, and MCP-mediated attacks~\citep{janjusevic2025hiding}. Surveys and audits highlight systematic coverage gaps in current defences and scanners~\citep{yu2025trustllmagentsurvey,brokman2025insights}.

\paragraph{Defences and formalisation.}
Boundary filters and guardrails are largely stateless~\citep{inan2023llamaguard,rebedea2023nemo,zeng2024autodefense,ganon2025diesel}, while component-specific defences target isolated vectors~\citep{masoud2026sdrag,pathmanathan2025ragpart,tan2024revprag,yu2026defense,wallace2024instructionhierarchy,suo2025argus,ramakrishnan2025securing,syed2025toward}. TrustAgent~\citep{hua2024trustagent} is closest in spirit but does not maintain a belief state; we formalise agentic security as a POMDP~\citep{kaelbling1998pomdp,chatterjee2012adversarialpomdp,gmytrasiewicz2005ipomdp} and operationalise belief-conditioned intervention.
Our ablation isolates the impact of statefulness and exposes a regime where scaling stateless checkpoints yields no marginal benefit under high correlation (Proposition~\ref{prop:correlation}).
Viewed this way, \system functions as a belief integrator that can layer atop existing component defences to supply cross-stage context.

\paragraph{Stateless scaling and trajectory detectors.}
Absent a shared information state, multi-point stateless deployments remain memoryless detectors whose block decisions depend only on the local observation. Formally, this corresponds to $P(\text{block}\mid o_t, h_t) = P(\text{block}\mid o_t)$, and Proposition~\ref{prop:correlation}(i) identifies a high-correlation regime where additional stateless checkpoints yield zero marginal benefit. Reactive trajectory-level methods such as Trajectory Guard~\citep{advani2026trajectory} detect anomalous completed interactions, whereas \system enforces proactive gating at internal trust boundaries via belief-conditioned intervention.

\paragraph{Control-theoretic foundations.}
POMDPs and their adversarial and interactive variants formalise acting under partial observability~\citep{kaelbling1998pomdp,chatterjee2012adversarialpomdp,gmytrasiewicz2005ipomdp}, but exact solvers are intractable for natural-language state spaces. Our approach operationalises approximate belief tracking via structured LLM inference and checkpointed intervention in the agent loop.

\section{Conclusion}
\label{sec:conclusion}

We introduced \system, formalising agentic RAG security as a POMDP where adversarial intent is latent and tracked via an approximate belief state. On 43,774 ART-SafeBench instances, \system achieves a $6.50\times$ average ASR reduction factor with negligible utility impact, and factorial ablation validates the predicted role of memory and checkpoint coverage.
On B1, the factorial ablation shows that belief state and multi-stage coverage are individually necessary (26.4\,pp and 13.6\,pp gaps) and jointly super-additive.
Because the trust layer is model-agnostic and inference-aligned, it can be deployed with specialised safety models or upgraded foundations, with defence fidelity improving with trust-model capability.
More broadly, the checkpoint set is configurable and should be selected to match application trust boundaries rather than a fixed pipeline.

\paragraph{Limitations.}
We rely on a deterministic GPT-4o judge, so evaluator bias may remain despite $>95\%$ manual agreement. Tool-flip attacks reveal a language-level observational-equivalence limit and require deterministic tool governance beyond semantic belief tracking. The overlay adds a $3.34\times$ latency overhead (Appendix~\ref{app:latency_cost}) and we do not evaluate adaptive checkpoint scheduling or parallelisation. As a fixed policy, \system may be vulnerable to white-box prompt optimisation and recursive injection, and other orchestration or memory designs may require different checkpoint placements.

\newpage
\bibliography{custom}

\appendix

\section{MTA Prompt Template}
\label{app:prompt}

The trust policy $\pi_T$ is encoded as a structured system prompt provided to the \mta's frozen LLM at each checkpoint. This prompt serves as the \emph{inference-time operationalisation} of the POMDP-motivated policy, comprising four functional components:

\textbf{(1) Role Specification:} Establishes the LLM as the decision-making kernel, enforcing conservative reasoning principles (precautionary principle, stateful assessment, and proportional response).

\textbf{(2) Checkpoint-Specific Directives:} Modular instructions tailored to the current observation type. For instance, C3 (Retrieval) focuses on identifying latent instruction overrides embedded in documents, while C2 (Action) validates the semantic consistency of tool calls against the user's original intent.

\textbf{(3) Belief State Injection:} The cumulative belief state $\Phi_{t-1}$ is injected verbatim into a designated ``Current Suspicion Context'' field. This provides the LLM with explicit access to the cross-stage history $h_{t-1}$, enabling the detection of distributed attacks.

\textbf{(4) Output Schema:} A strict JSON schema requiring: (a) risk assessment with chain-of-thought reasoning, (b) a scalar risk score $r \in [0,1]$, (c) the decision $\delta_t \in \{\approve, \mitig, \refuse\}$, (d) sanitised artifacts (if $\delta_t = \mitig$), (e) the updated belief state $\Phi_t$, and (f) internal justification for audit logs.

\paragraph{Verbatim Template (Abridged).}
The system prompt below illustrates how the POMDP policy is operationalised via natural language instructions.

\medskip
\noindent\fbox{\begin{minipage}{0.96\columnwidth}\footnotesize
\textbf{System Prompt $\pi_T$}\\[3pt]
You are the Modular Trust Agent (MTA) operating within a multimodal agentic RAG framework. The system operates in an adversarial environment where user queries, database entries, or tool outputs may contain malicious instructions, prompt-injection attempts, or other adversarial behaviour. Suspicious or harmful patterns can appear gradually across multi-turn or multi-agent conversations.

\medskip\noindent
\textbf{Role.} At each inference step, you receive:
(1)~Newly presented content $C$ (user query, tool decision, retrieved data, or final answer draft);
(2)~A partial Suspicion Context $\Phi$, your memory of suspicious or adversarial clues accumulated so far;
(3)~Policy guidelines $\pi_T$ for handling disallowed content.

\medskip\noindent
\textbf{Procedure.}
1.~Integrate $C$ with $\Phi$ to update beliefs about adversarial patterns. (Reflect partial observability: you do NOT see an entire malicious plan at once.)
2.~Decide among \{\textsc{Approve}, \textsc{Mitigate}, \textsc{Refuse}\}:
   - \textsc{Approve} if safe with no suspicious aspects.
   - \textsc{Mitigate} if partially disallowed but sanitisable while preserving valuable information.
   - \textsc{Refuse} if irredeemably harmful or disallowed (e.g., instructions for violent wrongdoing, persistent prompt-injection).
3.~Return structured JSON: \texttt{\{"decision": ..., "updatedContext": ..., "hiddenJustification": ...\}}.
\texttt{updatedContext} is the revised suspicion context; escalate suspicion logs upon repeated suspicious attempts. \texttt{hiddenJustification} is the internal rationale for audit logs.

\medskip\noindent
\textbf{Policy Rules.}
- If you see potential prompt injections (``Ignore your rules'', etc.), do not comply; use \textsc{Mitigate} for minor cases or \textsc{Refuse} if severe.
- If queries request disallowed content (violence instructions, extremist content, PII), \textsc{Refuse}.
- If suspicious patterns recur, escalate from \textsc{Approve} to \textsc{Mitigate} or \textsc{Refuse}.
- For hidden or obfuscated instructions to override the system, \textsc{Mitigate} or \textsc{Refuse}.
- Remember we are in a POMDP: carry forward suspicious signals in \texttt{updatedContext}.
\end{minipage}}

\medskip
\noindent\fbox{\begin{minipage}{0.96\columnwidth}\footnotesize
\textbf{Per-Checkpoint User Message}\\[3pt]
\texttt{CURRENT SUSPICION CONTEXT ($\Phi$):}\\
\texttt{\{$\Phi_{t-1}$\}}\\[2pt]
\texttt{NEW CONTENT TO CHECK:}\\
\texttt{\{$C_t$\}}\\[2pt]
\texttt{INSTRUCTIONS:}\\
\texttt{1. Integrate the new content with the suspicion context.}\\
\texttt{2. Decide among \{APPROVE, MITIGATE, REFUSE\}.}\\
\texttt{3. Return structured JSON.}
\end{minipage}}

\section{MTA Inference-Action Algorithm}
\label{app:algorithm}

Algorithm~\ref{alg:mta} formalises the control loop executed at each checkpoint.

\begin{algorithm}[t]
\small
\caption{\mta Inference-Action Cycle at Checkpoint $t$}
\label{alg:mta}
\begin{algorithmic}[1]
\Require Artifact $C_t$, belief state $\Phi_{t-1}$, policy $\pi_T$
\Ensure Decision $\delta_t$, updated $\Phi_t$, output $C'_t$
\State $o_t \leftarrow \textsc{Observe}(C_t, \mathrm{type}_t)$
\State $r_t, \Phi_t, \mathrm{trace}_t \leftarrow \textsc{Infer}(\Phi_{t-1}, o_t, \pi_T)$
\State $\delta_t \leftarrow \textsc{Decide}(r_t)$
\If{$\delta_t = \approve$}
    \State $C'_t \leftarrow C_t$
\ElsIf{$\delta_t = \mitig$}
    \State $C'_t \leftarrow \textsc{Sanitise}(C_t, \mathrm{trace}_t)$
\Else
    \State $C'_t \leftarrow \textsc{SafeResponse}()$; \textbf{terminate}
\EndIf
\State \Return $(\delta_t, \Phi_t, C'_t)$
\end{algorithmic}
\end{algorithm}

\section{Complete Pipeline with \mta}
\label{app:fullalg}

Algorithm~\ref{alg:full} presents the end-to-end processing loop, detailing the integration of our evaluated instantiation (C1-C5) into the agentic RAG workflow.

\begin{algorithm}[t]
\small
\caption{Agentic RAG Pipeline with \mta Integration}
\label{alg:full}
\begin{algorithmic}[1]
\Require Query $q$, knowledge store $\mathcal{K}$, tools $\mathcal{T}$
\Ensure Safe response $r$ or refusal
\State $\Phi_0 \leftarrow \varnothing$
\State $\delta_1, \Phi_1, q' \leftarrow \textsc{MTA}(q, \Phi_0)$ \Comment{C1: Query}
\If{$\delta_1 = \refuse$}
    \State \Return \textsc{Safe}()
\EndIf
\State $a \leftarrow \Amain.\textsc{Plan}(q')$
\State $\delta_2, \Phi_2, a' \leftarrow \textsc{MTA}(a, \Phi_1)$ \Comment{C2: Action}
\If{$\delta_2 = \refuse$}
    \State \Return \textsc{Safe}()
\EndIf
\State $D \leftarrow \textsc{Execute}(a', \mathcal{K}, \mathcal{T})$
\State $\delta_3, \Phi_3, D' \leftarrow \textsc{MTA}(D, \Phi_2)$ \Comment{C3: Data}
\If{$\delta_3 = \refuse$}
    \State \Return \textsc{Safe}()
\EndIf
\State $\Phi_4 \leftarrow \Phi_3$
\ForAll{tool output $t_i \in D'$} \Comment{C4: Tools}
    \State $\delta_4^i, \Phi_4, t'_i \leftarrow \textsc{MTA}(t_i, \Phi_4)$
    \If{$\delta_4^i = \refuse$}
        \State \Return \textsc{Safe}()
    \EndIf
\EndFor
\State $r \leftarrow \Lgen.\textsc{Generate}(q', D')$
\State $\delta_5, \Phi_5, r' \leftarrow \textsc{MTA}(r, \Phi_4)$ \Comment{C5: Output}
\If{$\delta_5 = \refuse$}
    \State \Return \textsc{Safe}()
\EndIf
\State \Return $r'$
\end{algorithmic}
\end{algorithm}

\section{Per-LLM Results on B1}
\label{app:perllm}

\begin{table}[t]
\centering
\footnotesize
\setlength{\tabcolsep}{5pt}
\renewcommand{\arraystretch}{1.08}
\begin{tabularx}{\columnwidth}{@{}Xccc@{}}
\toprule
\textbf{LLM Backbone} & \textbf{Base} & \textbf{+\,\mta} & \textbf{Red.} \\
\midrule
GPT-4o        & 76.5 & 15.7 & $4.9\!\times$ \\
GPT-4o-mini   & 78.4 & 31.5 & $2.5\!\times$ \\
Llama-3.3-70B & 70.9 &  7.5 & $9.5\!\times$ \\
GPT-4.1       & 78.0 & 14.0 & $5.6\!\times$ \\
\bottomrule
\end{tabularx}
\caption{Per-LLM end-to-end ASR~(\%) and reduction factors for B1 (Text Poisoning), varying the protected backbone while keeping the trust layer fixed.}
\label{tab:perllm_b1}
\end{table}

\section{ART-SafeBench Generation Framework}
\label{app:artsafebench_generation}

ART-SafeBench is constructed via a closed-loop validated synthesis pipeline that produces adversarial instances together with explicit success predicates.

\paragraph{Four-stage summary.}
Given an attack goal $g$ (e.g., OWASP-aligned objective) and an attacker capability class (e.g., black-box prompting vs.\ partial-access poisoning), the pipeline:
\textbf{(1) Specifies} a success predicate $\kappa_g$ that operationalises the violation of interest (harmful output, control subversion, exfiltration, or resource abuse);
\textbf{(2) Crafts} an adversarial stimulus $x$ and (optionally) benign context $\mu$ that realises the capability model;
\textbf{(3) Executes} a reference unhardened agentic RAG pipeline on $(x,\mu)$ to obtain an outcome trace $y$ (including intermediate tool and retrieval artifacts); and
\textbf{(4) Validates} the instance by checking whether $y$ satisfies $\kappa_g$, retaining only validated instances for which the reference system exhibits a violation under the chosen predicate.

\paragraph{Five-step instantiation.}
In our implementation, the pipeline is instantiated as: \textbf{(i) Definition} (sample $(g,\Gamma)$ and emit $\kappa_g$), \textbf{(ii) Crafting} (generate candidate $x$ compatible with $\Gamma$), \textbf{(iii) Context generation} (optionally generate $\mu$ for realism), \textbf{(iv) Simulator} (run the reference system deterministically to obtain $y$), and \textbf{(v) Judge} (deterministically evaluate $\kappa_g(y)$ and accept only if successful). Empirically, the closed-loop acceptance rate is $\sim$37\%.

\section{ART-SafeBench Record Schema}
\label{app:schema}

Each ART-SafeBench record follows a strict schema to ensure reproducibility:
\texttt{id} (SHA-256 hash of provenance),
\texttt{benchmark} (B1-B5),
\texttt{attack\_goal} (natural language description),
\texttt{attack\_payload} (the specific text/image/tool-call injection),
\texttt{benign\_context} (ground truth context),
\texttt{category} (OWASP classification), and
\texttt{metadata} (target model, random seed, timestamp).
Records are stored in JSON Lines format.

\section{Extended OWASP Results}
\label{app:owasp}

\begin{table}[t]
\centering
\footnotesize
\setlength{\tabcolsep}{5pt}
\renewcommand{\arraystretch}{1.05}
\begin{tabularx}{\columnwidth}{@{}Xrcc@{}}
\toprule
\textbf{OWASP Category} & $n$ & \textbf{Base} & \textbf{+\,\mta} \\
\midrule
Prompt injection (LLM01)  & 1,477 & 64 &  8 \\
Insecure output (LLM05)   & 716 & 71 & $<$1 \\
Sensitive info.\ (LLM02)  & 715 & 36 &  0 \\
Supply chain (LLM03)      & 623 & 58 &  0 \\
Misinfo./harmful (LLM09)  & 738 & 63 &  5 \\
Model DoS (LLM10)         & 709 & 61 &  7 \\
\bottomrule
\end{tabularx}
\caption{Per-OWASP ASR~(\%) for B3 for the six instantiated categories ($n\!=\!4{,}978$ of 10,005). Information disclosure (LLM02) and supply-chain (LLM03) attacks are eliminated entirely.}
\label{tab:owasp}
\end{table}

\section{Proof Sketches}
\label{app:proof_sketches}

\paragraph{Proposition~\ref{prop:memory}.}
The weak inequality follows from set inclusion: $\Pi_\Omega \subset \Pi_\Phi$.
For strictness, partial observability induces aliasing: there exist histories $h_t, h'_t$ with the same instantaneous observation $o_t$ but different posteriors over $\Sadv$. Any stateless policy must choose the same action on both $h_t$ and $h'_t$, while a belief-conditioned policy can separate them via $\Phi_t$, yielding strictly higher expected value under $J$ in general.

\paragraph{Proposition~\ref{prop:correlation}.}
(i) Under the stated condition, $d(o_t) = d(o_1)$ for all $t$, so $D_t = D_1$ deterministically and the union reduces to a single event. (ii) The belief state acts as an integrator: sub-threshold evidence at $t$ elevates $\Phi_t$, lowering the effective decision boundary at $t\!+\!1$ and creating a monotone increase in $P(D_{t+1}\!=\!1 \mid \cdots)$ beyond the memoryless baseline.

\section{Latency and Cost Analysis}
\label{app:latency_cost}

The \mta introduces $3.34\times$ latency overhead: mean end-to-end response time increases from $2.59 \pm 0.77$\,s to $8.64 \pm 2.92$\,s ($N\!=\!100$ NQ queries).
Each defended query traverses up to $K$ checkpoint LLM calls in addition to the baseline generation call ($K\!=\!5$ in our evaluated instantiation), yielding an estimated $\sim$3,500-5,000 additional input tokens and $\sim$800-1,200 output tokens per query (from belief-state context, checkpoint prompts, and structured JSON responses).
The overhead is dominated by the sequential dependency across checkpoints (C1-C5); the increased variance stems from the variable number and complexity of LLM calls conditioned on the input and evolving belief state.
Unvalidated optimisation paths include checkpoint parallelisation (C3/C4 concurrent), adaptive activation (low-risk queries bypass intermediate checkpoints), and asymmetric model selection (smaller LLM for low-risk checkpoints).
In security-sensitive deployments where the cost of a successful attack substantially exceeds latency cost, this overhead represents an acceptable trade-off.

\end{document}